\documentclass[letterpaper, 10 pt, conference]{ieeeconf} 
\IEEEoverridecommandlockouts                            
\usepackage{array}
\usepackage{amsmath}
\usepackage{amssymb}
\usepackage{mathtools}
\usepackage{url}
\usepackage{graphicx}
\usepackage[linesnumbered,ruled]{algorithm2e}
\usepackage[usenames,dvipsnames]{color}
\usepackage{xcolor}
\usepackage{float}
\usepackage[english]{babel}
\usepackage{bm}
\usepackage{tikz}
\usepackage[flushleft]{threeparttable}
\usetikzlibrary{matrix}
\usepackage[font={small,it}]{caption}
\usepackage{subcaption}
\newcommand{\RNum}[1]{\uppercase\expandafter{\romannumeral #1\relax}}
\newtheorem{theorem}{Theorem}[section]

\newtheorem{lemma}[theorem]{Lemma}
\newtheorem{definition}[theorem]{Definition}

\newtheorem{example}[theorem]{Example}

\title{Hierarchical $\mathcal{H}_{2}$ Control of Large-Scale Network Dynamic Systems}

\author{Nan~Xue~and~Aranya~Chakrabortty
\thanks{N. Xue and A. Chakrabortty are with the Department
of Electrical and Computer Engineering, North Carolina State University, Raleigh,
NC, 27695 USA, e-mail: nxue@ncsu.edu, achakra2@ncsu.edu }
\thanks{The work is supported partly by the US National Science Foundation (NSF) under grant ECCS 1054394.}
}

\begin{document}

\maketitle

\begin{abstract}

Standard $\mathcal{H}_{2}$ optimal control of networked dynamic systems tend to become unscalable with network size. Structural constraints can be imposed on the design to counteract this problem albeit at the risk of making the solution non-convex. In this paper, we present a special class of structural constraints such that the $\mathcal{H}_{2}$ design satisfies a quadratic invariance condition, and therefore can be reformulated as a convex problem. This special class consists of structured and weighted projections of the input and output spaces. The choice of these projections can be optimized to match the closed-loop performance of the reformulated controller with that of the standard $\mathcal{H}_{2}$ controller. The advantage is that unlike the latter, the reformulated controller results in a hierarchical implementation which requires significantly lesser number of communication links, while also admitting model and controller reduction that helps the design to scale computationally. We illustrate our design with simulations of a $500$-node network.

\end{abstract}
{\keywords Optimal Control, Quadratic Invariance, Network Systems, Controller Reduction.}

\IEEEpeerreviewmaketitle

\section{Introduction}

Driven by recent advances in cyber-physical systems, control synthesis for large-scale network dynamic systems has become an increasingly prevalent topic in the control and network communities \cite{nantac}. Physical networks such as power systems, wireless sensor networks, or the recently emerging internet-of-things consist of numerous heterogeneous subsystems that may be defined across complex topologies and wide geographical spans. The typical numbers of subsystems in these networks can scale from thousands to millions, making the design of tractable control mechanisms very challenging. To tackle the curse of dimension, traditional model reduction-based techniques such as singular perturbations \cite{singular}, balanced truncation \cite{moore}, and controller reduction \cite{obinata} have been developed decades ago to facilitate the analysis and control of large-scale systems, in general. However, these techniques are not readily applicable to networks as their control schemes are mostly unstructured, and hence agnostic of the network topology that contains important constraints for both design and implementation. One such pertinent constraint is communication. For example, conventional feedback controllers such as $\mathcal{H}_{2}$ optimal controllers are defined over unstructured dense transfer matrices, and thus implementing this control would necessitate an impractically large number of communication links for even a moderately sized network.

Starting from the idea of decentralized control \cite{siljak}, several seminal papers have developed design tools to impose desired implementation structure on controllers. One of them, the work in \cite{qinv}, states that a decentralized optimal control admits a convex reformulation if the structural constraint on the controller is quadratic invariant (QI) under the plant model. Built around quadratic invariance, papers including \cite{qinv}, \cite{lamperski} and \cite{sabau} incorporate sparsity structures in the controller, aiming to reduce the communication density and to cope with delays. The restriction of these designs is that the sparsity pattern meeting the QI condition is highly dependent on the sparsity of the open-loop network, and therefore, the choice of sparsity can be limited. Moreover, the design complexities of these controllers are dictated by the order of the open-loop plant, and thus can become unscalable for very large-scale networks. A suboptimal design is proposed in \cite{sparse} to generalize the choice of sparsity structure by relaxation algorithms using $l_{1}$-weighted norm. The design, however, is still computationally demanding.

Motivated by these challenges, in this paper we present a special class of structural constraint for designing $\mathcal{H}_{2}$ controllers. The constraint set is defined by two projections over the input and output spaces, which result in a hierarchical implementation of the controller. The implementation mechanism works as follows. Selected network subsystems send their output measurements to a designated set of coordinators. These coordinators take the average of these measurements, and exchange the information between themselves to generate a $\mathcal{H}_{2}$ control law. Each coordinator, thereafter, broadcasts this control signal to its respective set of subsystems. The overall execution, therefore, is completely hierarchical, and requires lesser number of communication links. Preliminary work incorporating this hierarchy has been proposed in our recent paper \cite{nanacc} for a LQG controller built upon the idea of clustering based projection. However, the design there is restricted by the choice of certain design parameters that guarantee closed-loop stability. In this paper we formalize the design in \cite{nanacc} into a general $\mathcal{H}_{2}$ framework, where such restrictions do not exist. The hierarchical structure is proven to meet the QI condition irrespective of the structure of the plant, as a result of which, the control problem is reformulated as a standard unconstrained $\mathcal{H}_{2}$ control. The reformulated problem preserves the hierarchical structure inside its input and output matrices, and thus allows for model and controller reduction techniques to reduce the design complexity without breaking the structure of the controller. The reduction technique adopted in this paper is Hamiltonian-based approximation for solving the underlying algebraic Riccati equations (ARE), as proposed in \cite{are1}. The recent paper \cite{local} also addresses similar goals as ours using receding-horizon control, but the scalability of their controller is subject to the sparsity structure of the open-loop network. Related hierarchical designs have been proposed in \cite{hier1,hier2}, however, they do not exploit the convex reformulation provided by quadratic invariance.

The remainder of this paper is organized as follows. Section \RNum{2} defines the hierarchical constraint, and formulates the problem of hierarchical $\mathcal{H}_{2}$ control. The same problem is parameterized and reformulated into an unconstrained $\mathcal{H}_{2}$ control via quadratic invariance in Section \RNum{3}. Section \RNum{4} presents a Hamiltonian-based reduction technique to simplify the design of the reformulated controller. The optimality gap between the hierarchical $\mathcal{H}_{2}$ and unconstrained $\mathcal{H}_{2}$ problems is discussed in Section \RNum{5}, where we also propose a design on clustering sets to tighten the gap. Validations of our proposed designs are illustrated via simulations of large-scale networks in Section \RNum{6}. Section \RNum{7} concludes the paper and presents some future works.

{\bf{Notation\ }} 
The following notations will be used throughout this paper: $|\mathcal{I}|_{c}$: cardinality of a set $\mathcal{I}$, $diag(m)$: diagonal matrix with vector $m$ on its principal diagonal, $diag(M,N)$: block-diagonal matrix with matrices $M$ and $N$ on its principal diagonal, $tr(M)$: trace operation on a matrix $M$, $\| M \|_{F}$: Frobenius norm of a matrix $M$, i.e. $\| M \|_{F}=\sqrt{tr(MM^{T})}$. A transfer function matrix (TFM) is defined as $g(s)=C(sI-A)^{-1}B+D$, with a realization form of $g(s)=\left[
\begin{array}{c|c}
A & B \\ \hline
C & D
\end{array}
\right]
$. $\mathbb{R}_{sp}$ and $\mathbb{R}_{p}$ denote respectively the set of real-rational strictly proper TFMs and the set of real-rational proper TFMs. $\mathcal{RH}_{\infty}$ denotes the set of all stable TFMs from $\mathbb{R}_{p}$, and $\mathcal{RH}_{2}$ denotes the set of all stable TFMs from $\mathbb{R}_{sp}$. The $\mathcal{H}_{\infty}$ norm of $g(s) \in \mathcal{RH}_{\infty}$ is defined by $\|g(s)\|_{\mathcal{H}_{\infty}}=sup_{\omega}\ \bar{\sigma}[ g(j\omega)] $, where $\bar{\sigma}$ denotes the largest singular value. The $\mathcal{H}_{2}$ norm of $g(s) \in \mathcal{RH}_{2}$ is defined by $\|g(s)\|_{\mathcal{H}_{2}}=\sqrt{\int_{-\infty}^{\infty}tr[g^{*}(t)g(t)]\mathrm{d}t}=\sqrt{\frac{1}{2\pi}\int_{-\infty}^{\infty}tr[g^{*}(j\omega)g(j\omega)]\mathrm{d}\omega}$.

Furthermore, a Lyapunov equation is a matrix equation in the form of 
\begin{align*}
A\Phi + \Phi A^{T} + BB^{T} = 0,
\end{align*}
and an algebraic Riccati equation (ARE) is written as
\begin{align*}
A^{T}X + X A + C^{T}C + XBD^{-1}B^{T}X = 0.
\end{align*}
We use $\Phi = \mathrm{LYAP}(A,B)$ and $X = \mathrm{ARE}(A,B,C,D)$ to respectively denote solutions from these two equations.

\section{Problem Formulation}

We motivate our design from the standard formulation of $\mathcal{H}_{2}$ optimal control. Consider a transfer function matrix $G(s)$ in a realization form 
\begin{align}
G(s) = \begin{bmatrix}
G_{11} & G_{12} \\ G_{21} & G_{22}
\end{bmatrix} = \left[
\begin{array}{c|cc}
A & B_{1} & B_{2} \\ \hline
C_{1} & 0 & D_{12} \\
C_{2} & D_{21} & 0
\end{array}
\right], \label{full}
\end{align}
where $G_{22}$ is the plant model with $A\in \mathbb{R}^{n\times n}$, $B_{2}\in \mathbb{R}^{n\times n_{u}}$ and $C_{2}\in \mathbb{R}^{n_{y}\times n}$. The plant $G_{22}$ is defined over $n_{s}\leq n$ interconnected subsystems with each modeled by
\begin{equation}
\begin{aligned}
\dot{x}_{i} (t) & = A_{ii}x_{i}(t) + \sum_{i\neq j}^{n_{s}} A_{ij}x_{j}(t) + B_{2,ii} u_{i}(t) \\
y_{i}(t) & = C_{2,ii}x_{i}(t), \quad i=1,...,n_{s}
\end{aligned}\quad , \label{subsys}
\end{equation}
where $A_{ii}$, $B_{2,ii}$ and $C_{2,ii}$ are submatrices with compatible dimensions from $A$, $B_{2} = diag(B_{2,11},...,B_{2,n_{s}n_{s}})$ and $C_{2} = diag(C_{2,11},...,C_{2,n_{s}n_{s}})$. Note that for $n_{y},n_{u}\leq n$, $B_{2,ii}$ and $C_{2,ii}$ can be null matrices, that is, a subsystem can have no input or output. Consider a standard $\mathcal{H}_{2}$ optimal control for $G(s)$, we begin with the following assumptions.

A1. $(A,B_{2})$ is stabilizable and $(C_{2},A)$ is detectable

A2. $D_{21}D_{21}^{T} \succ 0$ and $D_{12}^{T}D_{12}\succ 0$ 

A3. $(A,B_{1})$ and $(C_{1},A)$ have no uncontrollable or unobservable modes on the imaginary axis  

A4. $D_{12}^{T}C_{1}=0$ and $B_{1}D_{21}^{T}=0$ \\
Assumptions (A1-A3) are standard assumptions to formulate $\mathcal{H}_{2}$ optimal control \cite{rc}. Assumption (A4) is made to simplify the notations for presenting the technical results. Our proposed design is still valid if this assumption is removed. Let $K(s) \in \mathbb{R}_{p}^{n_{u}\times n_{y}}$ denote the controller of interest. The lower linear fractional transformation (LFT) of $G(s)$ and $K(s)$ is defined by
\begin{align}
f(G,K) := G_{11} + G_{12}K(I-G_{22}K)^{-1}G_{21}.
\end{align}
With these notations, the constrained $\mathcal{H}_{2}$ optimal control problem considered in this paper can be posed as
\begin{equation}
\begin{aligned}
& \mathrm{minimize}
& & \| f(G,K)\|_{\mathcal{H}_{2}} \\
& \mathrm{subject\ to}
& & \text{$K$ stabilizes $G$} \\
& & & K\in \mathcal{S}
\end{aligned}, 
\label{op1}
\end{equation}
where $\mathcal{S} \subseteq \mathbb{R}_{p}^{n_{u}\times n_{y}}$ represents the subset of all strictly proper transfer function matrices that satisfy certain structural constraints. Note that when $\mathcal{S} = \mathbb{R}_{p}^{n_{u}\times n_{y}}$, (\ref{op1}) becomes the unconstrained $\mathcal{H}_{2}$ design. In practice, the unconstrained case usually yields a controller defined over a dense transfer function matrix. As a result, implementation of this controller for any large-scale network dynamic systems where $n$, $n_{u}$ and $n_{y}$ may scale up to millions, would necessitate an impractically large number of communication links. To bypass this challenge, in this paper we propose a hierarchical structure for the constraint $\mathcal{S}$, that can significantly simplify the implementation of $K(s)$. Another advantageous property of this structure is that it allows (\ref{op1}) to be reformulated as a convex optimization problem as will be unfolded in Section \RNum{3}. The basic mechanism behind the hierarchical formulation is the partitioning of subsystems into a set of non-overlapping groups or {\it clusters}. This grouping strategy can be arbitrary, and is not necessarily dictated by any system property. The constraint $\mathcal{S}$ for this hierarchy can then be formulated as follows. 
\begin{definition}
Given the index set of subsystems as $\mathcal{V} = \{ 1,...,n_{s} \}$, and an integer $r$, where $0<r\leq n_{s}$, define $r$ non-empty, distinct, and non-overlapping subsets of $\mathcal{V}$ respectively denoted as $\mathcal{I} = \{ \mathcal{I}_{1},...,\mathcal{I}_{r} \}$, such that $\mathcal{I}_{1} \cup ... \cup \mathcal{I}_{r} = \mathcal{V}$. We denote the collection of subsystems in $\mathcal{I}_{i}$, $i=1,...,r$ as a {\it cluster}.
\end{definition}
We next define two structured projection matrices $P_{u}$ and $P_{y}$ for clusters $\mathcal{I}_{1},...,\mathcal{I}_{r}$ as follows.
\begin{definition}
Given a non-zero vector $w_{u} \in \mathbb{R}^{n_{u}}$ and the input index set $\mathcal{V}_{u} {=} \{ 1,...,n_{u} \}$, define $r$ non-overlapping and non-empty subsets denoted as $\mathcal{I}^{u} = \{ \mathcal{I}^{u}_{1},...,\mathcal{I}^{u}_{r} \}$, such that $\mathcal{I}^{u}_{i}$ includes indices of all the inputs in cluster $\mathcal{I}_{i}$ for $i=1,...,r$. The input projection matrix $P_{u} \in \mathbb{R}^{r\times n_{u}}$ is defined by  
\begin{align}
P_{u,[ij]} := \begin{cases}
\frac{w_{u,[j]}}{\| w_{u,[\mathcal{I}_{i}^{u}]} \|_{2} }, \quad j\in \mathcal{I}^{u}_{i} \\
0,\quad \text{otherwise} 
\end{cases}, \label{Pdefe} 
\end{align}
where the vector $w_{u,[\mathcal{I}_{i}^{u}]}$ is non-zero. The output projection matrix $P_{y}$ is defined in the same way by vector $w_{y} \in \mathbb{R}^{n_{y}}$, output index set $\mathcal{V}_{y} {=} \{ 1,...,n_{y} \}$ and subsets $\mathcal{I}^{y} = \{ \mathcal{I}^{y}_{1},...,\mathcal{I}^{y}_{r} \}$.
\label{Pdef} 
\end{definition}

\begin{figure*}
    \centering
    \begin{subfigure}[t]{0.3\textwidth}
    \centering
        \includegraphics[width=0.9\columnwidth]{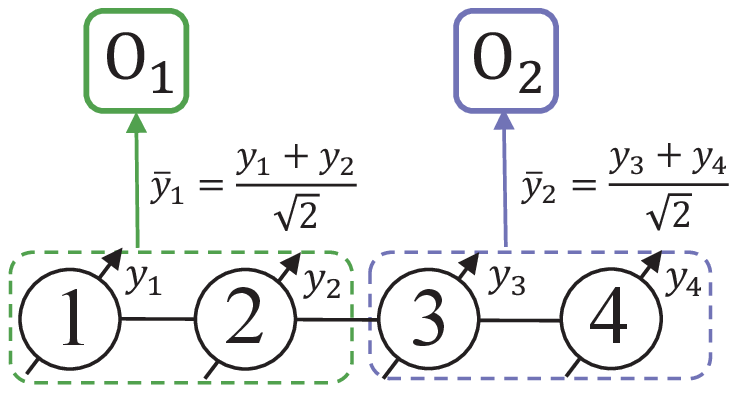}
        \caption{Step 1 - Output averaging}
        \label{cp1}
    \end{subfigure}
    ~        
    \begin{subfigure}[t]{0.3\textwidth}
    \centering
        \includegraphics[width=0.9\columnwidth]{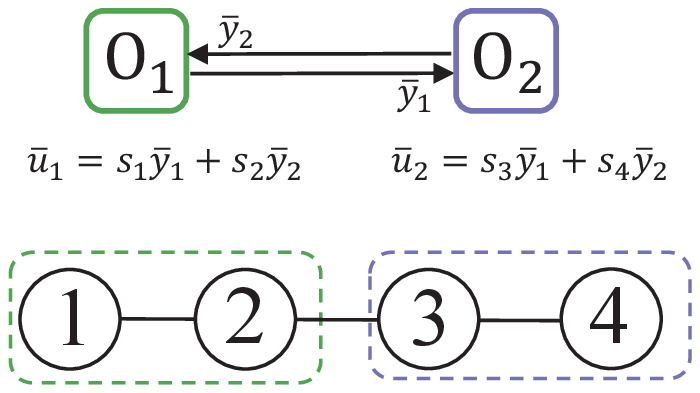}
        \caption{Step 2 - Lower-dimensional control}
        \label{cp2}
    \end{subfigure}
    ~        
    \begin{subfigure}[t]{0.3\textwidth}
    \centering
        \includegraphics[width=0.9\columnwidth]{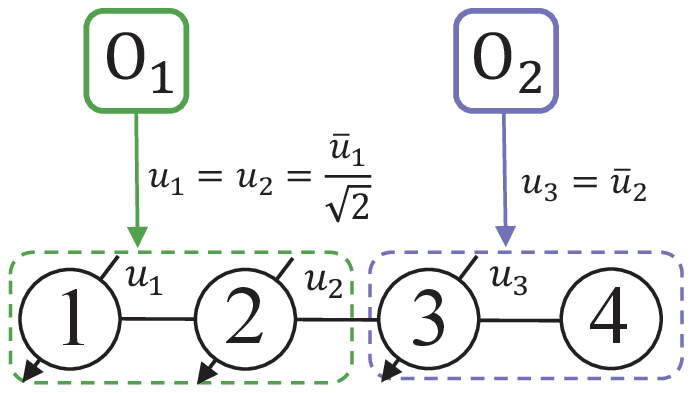}
        \caption{Step 3 - Control inversion}
        \label{cp3}
    \end{subfigure}
    \caption{Hierarchical architecture for implementing any controller in $\mathcal{S}$.}\label{cp}
\end{figure*}

Based on the projections $P_{u}$ and $P_{y}$, the constrained subspace $\mathcal{S}$ of all hierarchical controllers can be formally stated as follows.
\begin{definition}
The subspace $\mathcal{S}\in \mathbb{R}_{p}^{n_{u}\times n_{y}}$ admits a hierarchical structure defined over projection matrices $P_{u}\in \mathbb{R}^{r\times n_{u}}$ and $P_{y}\in \mathbb{R}^{r\times n_{y}}$ if there exists a lower-dimensional transfer matrix $\tilde{S}\in \mathbb{R}_{p}^{r\times r}$ such that
\begin{align}
\mathcal{S} = P_{u}^{T} \tilde{S} P_{y}.
\label{hstruct}
\end{align}

\end{definition}

The controllers in $\mathcal{S}$ contains the projection structures of $P_{y}$ and $P_{u}$, which lead to a sequential two-layer hierarchical control architecture. We illustrate this architecture along with the implementation of controllers in $\mathcal{S}$ by the following example. 
\begin{example}
Consider a networked system shown in the bottom layer of Fig. \ref{cp}, where four subsystems are indexed by $\mathcal{V}=\{ 1,2,3,4\}$, and are connected by a simple line graph. All four subsystems are assumed to have scalar outputs, and only the first three are equipped with inputs, i.e. $\mathcal{V}_{y}=\{ 1,2,3,4 \}$ and $\mathcal{V}_{u}=\{ 1,2,3 \}$. The system is partitioned into two clusters $\mathcal{I}_{1} = \{ 1,2 \}$ and $\mathcal{I}_{2} = \{ 3,4 \}$. Let both $w_{u}$ and $w_{y}$ to be vectors of all ones. Then given $\mathcal{I}_{1}^{u}=\{ 1,2\}$, $\mathcal{I}_{2}^{u}=\{ 3\}$, $\mathcal{I}_{1}^{y}=\{ 1,2\}$ and $\mathcal{I}_{2}^{y}=\{ 3,4\}$, the projection matrices can be constructed as
\begin{align*}
P_{u} = \begin{bmatrix}
\frac{1}{\sqrt{2}} & \frac{1}{\sqrt{2}} & 0 \\
0 & 0 & 1
\end{bmatrix},\ P_{y} = \begin{bmatrix}
\frac{1}{\sqrt{2}} & \frac{1}{\sqrt{2}} & 0 & 0 \\
0 & 0 & \frac{1}{\sqrt{2}} & \frac{1}{\sqrt{2}}
\end{bmatrix}.
\end{align*}
Denoting $\tilde{S} = \begin{bmatrix}
s_{1} & s_{2} \\
s_{3} & s_{4}
\end{bmatrix}$, the hierarchical subspace $\mathcal{S}$ follows as
\begin{align*}
\mathcal{S} = P_{u}^{T}\tilde{S}P_{y}  = \begin{bmatrix}
\frac{s_{1}}{2} & \frac{s_{1}}{2} & \frac{s_{2}}{\sqrt{2}} \\
\frac{s_{1}}{2} & \frac{s_{1}}{2} & \frac{s_{2}}{\sqrt{2}} \\
\frac{s_{3}}{2} & \frac{s_{3}}{2} & \frac{s_{4}}{\sqrt{2}} \\
\frac{s_{3}}{2} & \frac{s_{3}}{2} & \frac{s_{4}}{\sqrt{2}}
\end{bmatrix} . 
\end{align*}
Therefore, to implement the controller in $\mathcal{S}$, we designate two  coordinators, one for each cluster. These are denoted as $O_{1}$ and $O_{2}$ in Fig. \ref{cp}. The implementation of $\mathcal{S}$ can then be shown by the following three steps with illustrations in Fig. \ref{cp}.
\begin{itemize}
\item {\it Step 1 - Output averaging ($\bar{y} = P_{y}y$):} Each coordinator receives output measurements from its designated cluster, and computes the averaged output $\bar{y}_{i} = \sum_{j\in \mathcal{I}_{i}^{y}} \frac{w_{y,[j]} y_{j} }{\| w_{y,[\mathcal{I}_{i}^{y}]} \|_{2} }$ for $i=1,...,r$.
\item {\it Step 2 - Lower-dimensional control ($\bar{u} = \tilde{S}\bar{y}$):} Next, the coordinators communicate with each other, and exchange the averaged outputs $\bar{y}_{i}$, $i=1,...,r$. Each coordinator then computes a lower-dimensional $\mathcal{H}_{2}$ control input $\bar{u}_{i} = \tilde{S}_{i,:}\bar{y}$, $i=1,...,r$ in a distributed manner. 
\item {\it Step 3 - Control inversion ($u = P^{T}_{u}\bar{u}$):} In the final step, each coordinator broadcasts the control signal $u_{\mathcal{I}^{u}_{i}}$ to subsystems in its cluster by a simple scaling of $\bar{u}_{i}$, $i=1,...,r$.
\end{itemize}
\end{example}

By applying such a hierarchical structure, the total number of communication links required for implementing the control would be at most $n_{s}+\binom{r}{2}$. This can be much cheaper than $\binom{n_{s}}{2}$ required for the all-to-all communication in unstructured controllers. This structure also fits many common cyber-physical networks such as power systems where the subsystems may represent synchronous generators and loads, the clusters may represent utility companies, and coordinators may represent the control centers of each respective company. The proposed hierarchical implementation can also preserve data privacy between coordinators as only an averaged output $\bar{y}$ is shared between them. Therefore, no coordinator can infer the output measurements from subsystems assigned to other coordinators. 

Compared to \cite{nanacc}, where the LQG design is facilitated by only a single output projection matrix, the hierarchy in $\mathcal{S}$ is defined over both input and output projections. We next present the main results of this paper, starting with the convex reformulation for the constrained design (\ref{op1}).

\section{Convex Reformulation by Quadratic Invariance}
In this section, we present the convex reformulation for the hierarchical $\mathcal{H}_{2}$ problem (\ref{op1}) facilitated by the quadratic invariant property of $\mathcal{S}$. The following well-known results of Youla parameterization will be used to unfold the solution.
\begin{theorem}
\cite{rc} Let $F$ and $L$ be such that $A+LC_{2}$ and $A+B_{2}F$ are Hurwitz. Then all controllers that internally stabilize $G(s)$ can be parameterized by 
\begin{align}
K(s) = f(K_{nom},Q),
\end{align}
with any $Q \in \mathcal{RH}_{\infty}$ and $K_{nom}$ defined as
\begin{align}
K_{nom} = \left[
\begin{array}{c|cc}
A + B_{2}F + LC_{2} & -L & B_{2} \\ \hline
F & 0 & I \\
-C_{2} & I & 0
\end{array}
\right].
\label{nominal}
\end{align}
Furthermore, the closed-loop transfer function equals to
\begin{align}
f(G,K) = T_{11} + T_{12}QT_{21},
\label{para1}
\end{align}
where $T$ is given by
\begin{align}
\nonumber
& T = \begin{bmatrix}
T_{11} & T_{12} \\
T_{21} & T_{22}
\end{bmatrix} = \left[
\begin{array}{c|cc}
\hat{A} & \hat{B}_{1} & \hat{B}_{2} \\ \hline
\hat{C}_{1} & 0 & D_{12} \\
\hat{C}_{2} & D_{21} & 0
\end{array}
\right] \\
& \left[
\begin{array}{cc|cc}
A+B_{2}F & -B_{2}F & B_{1} & B_{2} \\
0 & A+LC_{2} & B_{1}+LD_{21} & 0 \\ \hline
C_{1}+D_{12}F & -D_{12}F & 0 & D_{12} \\
0 & C_{2} & D_{21} & 0
\end{array}
\right].
\label{para2}
\end{align}
\label{para}
\end{theorem}

Given $J$, $T_{11}$, $T_{12}$ and $T_{21}$ from Theorem \ref{para}, one can rewrite the original $\mathcal{H}_{2}$ problem (\ref{op1}) into a model matching problem with respect to the Youla parameter $Q$ as
\begin{equation}
\begin{aligned}
& \mathrm{minimize}
& & \| T_{11} + T_{12}QT_{21} \|_{\mathcal{H}_{2}}  \\
& \mathrm{subject\ to}
& & Q\in \mathcal{RH}_{\infty} \\
& & & K(s)=f(K_{nom},Q)\in \mathcal{S}
\end{aligned}\quad . 
\label{op2}
\end{equation}
Note that $K(s)=f(K_{nom},Q)\in \mathcal{S}$ is not an affine constraint in $Q$, in general. However, it has been shown in \cite{qinv} that $K \in \mathcal{S}$ is equivalent to the affine constraint $Q \in \mathcal{S}$ in Youla domain if the subspace $\mathcal{S}$ is quadratic invariant under the plant model $G_{22}$. We introduce the notion of quadratic invariance according to \cite{qinv} as follows.
\begin{definition}
The subspace $\mathcal{S}$ is called quadratic invariant under $G_{22}$ if $KG_{22}K \in \mathcal{S}$ for all $K \in \mathcal{S}$.
\end{definition}

From this definition, it can be easily verified that for any controller $K = P_{u}^{T}\tilde{K}P_{y} \in \mathcal{S}$, quadratic invariance holds for our hierarchical constraint $\mathcal{S}$ under $G_{22}$ given that 
\begin{align*}
KG_{22}K=P_{u}^{T}(\tilde{K}P_{y}G_{22}P_{u}^{T}\tilde{K})P_{y} \in \mathcal{S} 
\end{align*}
for $\tilde{K}P_{y}G_{22}P_{u}^{T}\tilde{K} \in \mathbb{R}_{sp}^{r\times r}$. As $K\in \mathcal{S}$ is equivalent to $Q\in \mathcal{S}$, (\ref{op2}) becomes
\begin{equation}
\begin{aligned}
& \mathrm{minimize}
& & \| T_{11} + T_{12}QT_{21} \|_{\mathcal{H}_{2}} \\
& \mathrm{subject\ to}
& & Q\in \mathcal{RH}_{\infty},\ Q\in \mathcal{S}
\end{aligned}\quad .
\label{op3}
\end{equation}
To this end, we can find the optimal solution for (\ref{op1}) as follows.
\begin{theorem}
If $(A,B_{2}P_{u}^{T})$ is stabilizable and $(P_{y}C_{2},A)$ is detectable, the hierarchical $\mathcal{H}_{2}$ problem (\ref{op1}) admits an optimal solution $K_{opt} = P_{u}^{T}\tilde{K}_{opt}P_{y}$,
\begin{align}
\tilde{K}_{opt} = \left[
\begin{array}{c|c}
A + B_{2}P_{u}^{T}F_{2} + L_{2}P_{y}C_{2} & -L_{2}  \\ \hline
F_{2} & 0
\end{array}
\right],
\label{Kopt}
\end{align}
where $R_{1} = P_{u}D_{12}^{T}D_{12}P_{u}^{T}$, $R_{2}=P_{y}D_{21}D_{21}^{T}P_{y}^{T}$, $F_{2} = -R_{1}^{-1}P_{u}B_{2}^{T}X$ and $L_{2} = -YC_{2}^{T}P_{y}^{T}R_{2}^{-1}$. Matrices $X = X^{T}\succeq 0$ and $Y=Y^{T}\succeq 0$ are unique solutions of AREs \footnote{Removing Assumption A4 will add extra coupled terms to $A$, $R_{1}$ and $R_{2}$. Interested readers can refer to \cite{rc} for the expressions. This parameter change, however, does not affect design procedures presented in this paper. }
\begin{align}
X & = \mathrm{ARE}(A,B_{2}P_{u}^{T},C_{1},R_{1}), \label{aremain} \\
Y & = \mathrm{ARE}(A^{T},C_{2}^{T}P_{y}^{T},B_{1}^{T},R_{2}). \label{aredual}
\end{align} 
\label{t1}
\end{theorem}
\begin{proof}
From (\ref{hstruct}), the constraint $Q\in \mathcal{S}$ implies that there exists $\tilde{Q}\in \mathbb{R}_{p}^{r\times r}$ such that $Q = P^{T}_{u}\tilde{Q}P_{y}$. By replacing $Q$ with $P^{T}_{u}\tilde{Q}P_{y}$, (\ref{op3}) becomes
\begin{equation}
\begin{aligned}
& \mathrm{minimize}
& & \| T_{11} + T_{12}P_{u}^{T}\tilde{Q}P_{y}T_{21} \|_{\mathcal{H}_{2}} \\
& \mathrm{subject\ to}
& & \tilde{Q}\in \mathcal{RH}_{\infty}
\end{aligned}.
\label{op4}
\end{equation}
Now consider an intermediate TFM $\bar{T}$ defined as 
\begin{align*}
& \bar{T} = \begin{bmatrix}
T_{11} & T_{12}P_{u}^{T} \\
P_{y}T_{21} & T_{22}
\end{bmatrix} = \\
& \left[
\begin{array}{cc|cc}
A+B_{2}F & -B_{2}F & B_{1} & B_{2}P_{u}^{T} \\
0 & A+LC_{2} & B_{1}+LD_{21} & 0 \\ \hline
C_{1}+D_{12}F & -D_{12}F & 0 & D_{12}P_{u}^{T} \\
0 & P_{y}C_{2} & P_{y}D_{21} & 0
\end{array}
\right].
\end{align*}
Given that $(A,B_{2}P_{u}^{T})$ is stabilizable and $(P_{y}C_{2},A)$ is detectable, there exists matrices $\bar{F}$ and $\bar{L}$ such that $A + B_{2}P_{u}^{T}\bar{F}$ and $A + \bar{L}P_{y}C_{2}$ are both Hurwitz. Therefore, by choosing $F$ and $L$ as $P_{u}^{T}\bar{F}$ and $\bar{L}P_{y}$ respectively in $\bar{T}$, the TFM within the $\mathcal{H}_{2}$ norm in (\ref{op4}) can be rewritten as the LFT form
\begin{align}
T_{11} + T_{12}P_{u}^{T}\tilde{Q}P_{y}T_{21} = f(\bar{G},\bar{K}), \label{LFTde}
\end{align}
where $\bar{G}$ is specified by 
\begin{align}
\bar{G} = \left[
\begin{array}{c|cc}
A & B_{1} & B_{2}P_{u}^{T} \\ \hline
C_{1} & 0 & D_{12}P_{u}^{T} \\
P_{y}C_{2} & P_{y}D_{21} & 0
\end{array}
\right]
\end{align}
and $\bar{K}$ is any stabilizing controller for $\bar{G}$ parameterized by Theorem \ref{para}. By (\ref{LFTde}), the problem (\ref{op4}) becomes an unconstrained $\mathcal{H}_{2}$ design for $\bar{G}$, which yields the optimal controller $\tilde{K}_{opt}$ as in (\ref{Kopt}). The optimal hierarchical controller, therefore, follows as $K_{opt} = P_{u}^{T}\tilde{K}_{opt}P_{y}$.
\end{proof}

Note that the hierarchical structure of $K_{opt}$ is defined by pre- and post-projections $P_{u}^{T}$ and $P_{y}$ over a lower-dimensional controller $\tilde{K}_{opt}$. This allows applications of conventional unstructured model reduction techniques for designing $\tilde{K}_{opt}$, without breaking the hierarchical structure of $K_{opt}$. We present one such reduction technique in the next section to simplify the design of $K_{opt}$.

\section{Hamiltonian-Based Approximation}
Constructing the optimal controller $K_{opt}$ requires two matrices $X$ and $Y$ from AREs (\ref{aremain}) and (\ref{aredual}). Given that $X$ and $Y$ are defined by the same ARE operator, we will present the results in this section based on $X$ only. We define the Hamiltonian matrix for (\ref{aremain}) as
\begin{align}
H {:=} \begin{bmatrix}
A & -M  \\
- C_{1}^{T}C_{1} & -A^{T}
\end{bmatrix},
\end{align}
where $M = B_{2}P_{u}^{T}R_{1}^{-1}P_{u}B_{2}^{T}$. The eigenvalues of $H$ are symmetric about the imaginary axis \cite{rc}. Suppose that $H$ is diagonalizable and that the columns of the matrix $\begin{bmatrix}
Z_{1} \\ Z_{2}
\end{bmatrix}_{2n\times n}$ span the stable invariant subspace of $H$, i.e.
\begin{align}
H \begin{bmatrix}
Z_{1} \\ Z_{2}
\end{bmatrix} = \begin{bmatrix}
Z_{1} \\ Z_{2}
\end{bmatrix} \Lambda^{-}, \label{hamieig}
\end{align}
where $\Lambda^{-} = diag([\lambda^{-}_{1},...,\lambda^{-}_{n}])$ consists of all the eigenvalues of $H$ in the left-half plane. The matrix $X$ can then be found by $X=Z_{2}Z_{1}^{-1}$ \cite{rc}. Computing the full stable eigenspaces above, however, is subject to $\mathcal{O}(n^{3})$ complexity, and thus can become unscalable for a large network. In order to make the hierarchical design applicable, we next consider the problem of finding an approximate and computationally simpler solution for $X$, noted as $\bar{X}$ (and $\bar{Y}$ for $Y$). 

We introduce a few notations based on the stable invariant subspace of $H$ defined in (\ref{hamieig}). Let $\Lambda_{\kappa}^{-} \in \mathbb{R}^{\kappa \times \kappa}$ be a principal submatrix of $\Lambda^{-}$, and denote the corresponding subspace partitions by $Z_{1\kappa}\in \mathbb{R}^{n\times \kappa}$ and $Z_{2\kappa}\in \mathbb{R}^{n\times \kappa}$, i.e. 
\begin{align*}
\Lambda^{-} = \begin{bmatrix}
\Lambda^{-}_{\kappa} & \\ & \Lambda^{-}_{\bar{\kappa}}
\end{bmatrix},\ Z_{1} = \begin{bmatrix}
Z_{1\kappa} & Z_{1\bar{\kappa}}
\end{bmatrix},\ Z_{2} = \begin{bmatrix}
Z_{2\kappa} & Z_{2\bar{\kappa}}
\end{bmatrix}.
\end{align*}
Using the lower-dimensional Hamiltonian subspaces $Z_{1\kappa}$ and $Z_{2\kappa}$, one can construct an approximate solution for (\ref{aremain}) as
\begin{align}
\bar{X} = Z_{2\kappa}(Z_{2\kappa}^{T}Z_{1\kappa})^{-1}Z_{2\kappa}^{T}.
\label{hamiapp}
\end{align}
This approximation form has been previously proposed in \cite{are1}, although the choice of the $\kappa$ eigenvalues contained in $\Lambda_{\kappa}^{-}$ was not determined. In contrast, motivated by conventional frequency-weighted model reduction, we decide the partitioning of $\Lambda^{-}$ from the minimization problem
\begin{equation}
\begin{aligned}
& \underset{\bar{X}}{ \mathrm{minimize} }
& & \| (X-\bar{X})W_{x}\|_{\mathcal{H}_{2}} 
\end{aligned},
\label{ophm}
\end{equation}
where $W_{x} = (sI - A + MX)^{-1}B_{1}$ is a frequency weight representing the closed-loop output response. Due to the non-convex nature of the model reduction, it is intractable to find an exact minimum for (\ref{ophm}). Therefore, we approach (\ref{ophm}) by an upper bound minimization derived as follows.

\begin{theorem}
Denote $E_{\kappa} {=} Z_{2\bar{\kappa}} {-} Z_{2\kappa}(Z_{2\kappa}^{T}Z_{1\kappa})^{-1}Z_{2\kappa}^{T}Z_{1\bar{\kappa}}$. The objective function in (\ref{ophm}) is upper bounded by
\begin{align}
\| (X-\bar{X})W_{x}\|_{\mathcal{H}_{2}} \leq \epsilon \| E_{\kappa}\|_{F}, \label{mainbound}
\end{align}
where $\epsilon := \sqrt{\sum_{i=\kappa + 1}^{n} C_{ii}}$, and $\mathcal{C} \in \mathbb{R}^{n\times n}$ is a Cauchy-like matrix defined by
\begin{align}
\mathcal{C}_{ij} := -\frac{1}{\lambda_{i}^{-}+\lambda_{j}^{-}} [Z_{1}^{-1}B_{1}B_{1}^{T}Z_{1}^{-T}]_{ij} . \label{cau}
\end{align} 
\label{theob}
\end{theorem}
\begin{proof}
See Appendix.
\end{proof}

From Theorem \ref{theob}, the norm of interest $\| (X-\bar{X})W_{x} \|_{\mathcal{H}_{2}}$ is linearly bounded by a scalar $\epsilon$. The minimization in (\ref{ophm}) can then be approached by minimizing $\epsilon$ with respect to the choice of $\Lambda^{-}_{\kappa}$. It is worth emphasizing that the value of $\mathcal{C}_{ii}$ is inversely proportional to the magnitude of the closed-loop eigenvalue $\lambda^{-}_{i}$ for $i=\kappa+1,...,n$. Therefore, one heuristic way for minimizing $\epsilon$ can be to order the eigenvalues of $\Lambda^{-}$ as $0< |\lambda^{-}_{1}| \leq ...\leq |\lambda^{-}_{n}|$, and let $\Lambda_{\kappa}^{-}$ contain the first $\kappa$ eigenvalues with the smallest magnitudes. In that case, $\lambda_{\kappa +1}^{-}$ till $\lambda_{n}^{-}$ will have larger magnitudes, and thus both $\mathcal{C}_{ii}$ and $\epsilon$ will be smaller positive numbers. Hence, computation of $\bar{X}$ will require only the eigenspace of these first $\kappa$ eigenvalues, which can be completed in $\mathcal{O}(n\kappa^{2})$ time using Krylov-subspace based techniques, i.e. Arnoldi algorithm \cite{matrix}. This complexity is more tractable than $\mathcal{O}(n^{3})$ of finding an exact solution $X$. The matching between $X$ and $\bar{X}$ becomes better as the gap between $\lambda^{-}_{\kappa}$ and $\lambda^{-}_{\kappa+1}$ increases. If
\begin{align}
0< |\lambda^{-}_{1}| \leq ...\leq |\lambda^{-}_{\kappa}| \ll |\lambda^{-}_{\kappa+1}| \leq ...\leq |\lambda^{-}_{n}|,
\end{align}
which means that closed-loop network is naturally clustered into $\kappa$ coherent groups \cite{nantac}, $\epsilon$ becomes negligible, indicating a close matching between $X$ and $\bar{X}$. The computational complexity for $\bar{X}$ can be further reduced if $\kappa$ is a sufficiently small number. 

The approximate solution $\bar{X}$, however, does not come without a drawback. In general, $\bar{X}$ can only stabilize $A-M\bar{X}$ over the subspace $Z_{1\kappa}$, and hence a small $\kappa$ will not necessarily stabilize an unstable plant. Therefore, we briefly comment on the stabilizability condition of $\bar{X}$ to conclude this section. First, some useful properties of $\bar{X}$ are summarized as follows. 
\begin{lemma}
\cite{are1} The approximate solution $\bar{X}$ given by (\ref{hamiapp}) satisfies $\bar{X}\preceq X$, and its residue can be written as
\begin{align}
\mathcal{R}(\bar{X}) := A^{T}\bar{X} + \bar{X}A + C_{1}^{T}C_{1} - \bar{X}M\bar{X} =\bar{C}^{T}_{1}\bar{C}_{1}, \label{areresi}
\end{align}
where $\bar{C}^{T}_{1} = [ I - Z_{2\kappa}(Z_{2\kappa}^{T}Z_{1\kappa})^{-1}Z_{1\kappa}^{T} ]C_{1}^{T}$.
\label{hprop}
\end{lemma}
A simple stability condition then follows immediately from Lemma \ref{hprop}.
\begin{theorem}
The approximate solution $\bar{X}$ is stabilizing if $C_{1}^{T}C_{1} - \bar{C}^{T}_{1}\bar{C}_{1} \succeq 0$, and $(C_{1}^{T}C_{1} - \bar{C}^{T}_{1}\bar{C}_{1},A)$ has no unobservable modes on the imaginary axis.
\label{stab}
\end{theorem}
\begin{proof}
By moving $\bar{C}_{1}^{T}\bar{C}_{1}$ to the left hand side of (\ref{areresi}), (\ref{areresi}) becomes an ARE in the form of $A^{T}\bar{X} + \bar{X}A + C_{1}^{T}C_{1} - \bar{C}^{T}_{1}\bar{C}_{1} - \bar{X}M\bar{X} = 0$. The stabilizability for $\bar{X}$ then follows from standard conditions in \cite{rc} based on $(C_{1}^{T}C_{1} - \bar{C}^{T}_{1}\bar{C}_{1},A)$.
\end{proof}

Theorem \ref{stab} implies that one sufficient condition for stabilizing an unstable system is to choose a positive-definite matrix $C_{1}^{T}C_{1}$ with a large norm. The flip side, however, is that this would lead to a larger control effort, thereby amplifying the measurement noise in the $\mathcal{H}_{2}$ design. An alternative is to gradually increase the value of $\kappa$, and check for closed-loop stability. Theorem \ref{stab} can thus be used as a stability test, requiring computation of only $Y_{\kappa}$ and $Z_{\kappa}$ without any need for checking the eigenvalues of $A-M\bar{X}$.

\section{Design of Hierarchical Constraint $\mathcal{S}$}

The results presented so far produces a stabilizing controller $K_{opt}$, which minimizes the values of $\| f(G,K)\|_{\mathcal{H}_{2}}$ for any given choice of $\mathcal{S}$ in (\ref{op1}). In this section, we show that it is also possible to design $\mathcal{S}$ such that this norm is close to the unconstrained $\mathcal{H}_{2}$ problem, that is the problem in (\ref{op1}) without the constraint $K \in \mathcal{S}$.

\subsection{Quantification of the Optimality Gap}



The unconstrained $\mathcal{H}_{2}$ problem can be written in model matching form as
\begin{equation}
\begin{aligned}
& \mathrm{minimize}
& & \| J_{1}(Q) \|_{\mathcal{H}_{2}} = \| T_{11} + T_{12}QT_{21} \|_{\mathcal{H}_{2}} \\
& \mathrm{subject\ to}
& & Q\in \mathcal{RH}_{\infty}
\end{aligned}\quad .
\label{opu}
\end{equation}
We start with the following two lemmas. 
\begin{lemma}
The optimal solution of (\ref{opu}) is given by 
\begin{align}
Q_{*} = - W_{L} \bar{W}_{R} = - \bar{W}_{L} W_{R},
\end{align}
where $W_{L}$, $\bar{W}_{L}$, $W_{R}$ and $\bar{W}_{R}$ are defined as
\begin{align*}
& W_{L} = \left[
\begin{array}{c|c}
\hat{A} + \hat{B}_{2}\hat{F} & I \\ \hline
\hat{F} & 0
\end{array}
\right], \quad \bar{W}_{L} = \left[
\begin{array}{c|c}
\hat{A} + \hat{B}_{2}\hat{F} & \hat{F} \\ \hline
\hat{B}_{2}\hat{F} & \hat{F}
\end{array}
\right], \\
& W_{R} = \left[
\begin{array}{c|c}
\hat{A} + \hat{L}\hat{C}_{2} & \hat{L} \\ \hline
I & 0
\end{array}
\right], \quad \bar{W}_{R} = \left[
\begin{array}{c|c}
\hat{A} + \hat{L}\hat{C}_{2} & \hat{L} \\ \hline
\hat{L}\hat{C}_{2} & \hat{L}
\end{array}
\right], \\
& \hat{F} = - (D_{12}^{T}D_{12})^{-1}\hat{B}_{2}^{T}\hat{X},\ \hat{X} = \mathrm{ARE}(\hat{A},\hat{B}_{2},\hat{C}_{1},D_{12}^{T}D_{12}), \\ 
& \hat{L} {=} - \hat{Y}\hat{C}_{2}^{T}(D_{21}D_{21}^{T})^{-1},\ \hat{Y} {=} \mathrm{ARE}(\hat{A}^{T},\hat{C}_{2}^{T},\hat{B}_{1}^{T},D_{21}D_{21}^{T}),
\end{align*}
and matrices $\hat{A}$, $\hat{B}_{1}$, $\hat{B}_{2}$, $\hat{C}_{1}$ and $\hat{C}_{2}$ are defined in (\ref{para2}).
\label{spect}
\end{lemma}
\begin{proof}
The solution for the model matching problem (\ref{opu}) follows from spectral factorization. See \cite{rc}. 
\end{proof}

\begin{lemma}
The constrained $\mathcal{H}_{2}$ problem (\ref{op3}) is equivalent to the unconstrained problem 
\begin{equation}
\begin{aligned}
& \mathrm{minimize}
& & \| J_{2}(Q) \|_{\mathcal{H}_{2}} = \| T_{12}P_{u}^{T}P_{u} Q P_{y}^{T}P_{y}T_{21} \\
& & & \quad\quad\quad\quad\quad\quad + T_{11} \|_{\mathcal{H}_{2}} \\
& \mathrm{subject\ to}
& & Q \in \mathcal{RH}_{\infty}
\end{aligned}.
\label{ope}
\end{equation}
\label{equiv}
\end{lemma}
\begin{proof}
See Appendix.
\end{proof}

With Lemma \ref{spect} and \ref{equiv}, we can quantify the optimality gap between the constrained and unconstrained $\mathcal{H}_{2}$ problems as follows. 
\begin{theorem}
Denote the optimal values of (\ref{opu}) and (\ref{ope}) by $J_{1*}$ and $J_{2*}$ respectively. The optimal value $J_{2*}$ is upper bounded by
\begin{align}
J_{2*}^{2} \leq J_{1*}^{2} + 2\xi J_{1*} + \xi^{2},\quad \xi = \epsilon_{1} \xi_{u} + \epsilon_{2} \xi_{y} + \epsilon_{2} \xi_{y}, \label{bound} 
\end{align}
where $\epsilon_{1} = \| T_{12}\|_{\mathcal{H}_{\infty}} \| T_{21}\|_{\mathcal{H}_{\infty}}\| \bar{W}_{R}\|_{\mathcal{H}_{\infty}}$ and $\epsilon_{2} = \| T_{12}\|_{\mathcal{H}_{\infty}} \| T_{21}\|_{\mathcal{H}_{\infty}} \| \bar{W}_{L} \|_{\mathcal{H}_{\infty}} $ are bounded scalars, and $\xi_{u}$ and $\xi_{y}$ are defined by
\begin{align*}
& \xi_{u} = \| (I {-} P_{u}^{T}P_{u})\hat{F}\Phi_{u}^{\frac{1}{2}}\|_{F}, \quad \xi_{y} = \| (I {-} P_{y}^{T}P_{y})\hat{L}^{T}\Phi_{y}^{\frac{1}{2}}\|_{F}, \\
& \Phi_{u} = \mathrm{LYAP}(\hat{A} + \hat{B}_{2}\hat{F}, I),\ \Phi_{y} = \mathrm{LYAP}((\hat{A} + \hat{L}\hat{C}_{2})^{T},I).
\end{align*}
The equality of (\ref{bound}) is attained at $\xi_{u} = \xi_{y} = 0$.
\label{boundr}
\end{theorem}
\begin{proof}
See Appendix.
\end{proof}

Smaller values of $\xi_{u}$ and $\xi_{y}$ will result in a smaller value of $\xi$, and therefore, in a tighter gap between $J_{1*}$ and $J_{2*}$. Also note that $\xi_{u}$ and $\xi_{y}$ are monotonic functions of $P_{u}$ and $P_{y}$, respectively. Therefore, as a relaxation we design $P_{u}$ and $P_{y}$ from 
\begin{equation}
\begin{aligned}
& \underset{P_{u}}{\mathrm{minimize}}
& & \xi_{u} 
\end{aligned} \quad \& \quad \begin{aligned}
& \underset{P_{y}}{\mathrm{minimize}}
& & \xi_{y} 
\end{aligned}. \label{xi}
\end{equation}

\subsection{Design of $P_{u}$ and $P_{y}$}

The projections $P_{u}$ and $P_{y}$ are defined over two variables, i.e. the clustering sets $(\mathcal{I}^{u}, \mathcal{I}^{y})$ and clustering weights $(w_{u},w_{y})$. 
Under the scope of this paper, we consider solving only the clustering sets $(\mathcal{I}^{u}, \mathcal{I}^{y})$ from (\ref{xi}), as these sets define the combinatorial structure of $\mathcal{S}$. To gain full degree of freedom in designing $(\mathcal{I}^{u}, \mathcal{I}^{y})$, we make the following choice of $(w_{u},w_{y})$ to satisfy the stabilizability and detectability conditions required for Theorem \ref{t1}.
\begin{lemma}
$(A,B_{2}P_{u}^{T})$ is stabilizable and $(P_{y}C_{2},A)$ is detectable if $w_{u} = V_{l}v_{u}$ and $w_{y} = V_{r}v_{y}$ for any $v_{u}$ and $v_{y}$ that satisfy
\begin{align}
V_{l}^{T}B_{2}B_{2}^{T}V_{l} v_{u} \neq \mathbf{0}, \quad V_{r}^{T}C_{2}^{T}C_{2}V_{r} v_{y} \neq \mathbf{0}, \label{pbhw}
\end{align}
where $\mathbf{0}$ is a vector with all zero entries. Matrices $V_{l}$ and $V_{r}$ denote the left and right eigenvector matrices such that $A^{T}V_{l} = V_{l}\Lambda_{+}$, $AV_{r} = V_{r}\Lambda_{+}$, and $\Lambda_{+}\succeq 0$ is a diagonal matrix that contains all unstable eigenvalues (including imaginary eigenvalues) of $A$. 
\label{lemmaw}
\end{lemma}

The proof is omitted as this lemma follows directly from PBH test. The implication of Lemma \ref{lemmaw} is that when the open-loop system (\ref{full}) is not stable, vectors $w_{u}$ and $w_{y}$ can be chosen to meet linear inequalities in (\ref{pbhw}), which can be done by linear programming. When the open-loop system (\ref{full}) is stable, the stabilizability and detectability conditions become trivial. In this case, one can choose $w_{u}$ and $w_{y}$ to be any non-zero vectors, e.g. a vector of all ones. 

Once $(w_{u},w_{y})$ are fixed, it has been shown in our recent paper \cite{nantac} that the minimization (\ref{xi}) with respect to $(\mathcal{I}^{u},\mathcal{I}^{y})$ is equivalent to an unsupervised clustering optimization. A simple yet efficient heuristic algorithm to solve this problem is weighted k-means \cite{kmeans}. Therefore, one can design the clustering sets by simply providing data matrices $(\hat{F}\Phi_{u}^{\frac{1}{2}},\hat{L}^{T}\Phi_{y}^{\frac{1}{2}})$, weight vectors $(w_{u},w_{y})$, and number of clusters $r$ as inputs to the k-means algorithm, i.e.
\begin{align}
\mathcal{I}^{u} = \mathrm{kmeans}(\hat{F}\Phi_{u}^{\frac{1}{2}}, w_{u}, r) \ \& \ \mathcal{I}^{y} = \mathrm{kmeans}(\hat{L}^{T}\Phi_{y}^{\frac{1}{2}}, w_{y}, r).
\label{kmeansop}
\end{align}

\section{Numerical Example}

\begin{figure}
    \centering
    \includegraphics[width=0.8\columnwidth]{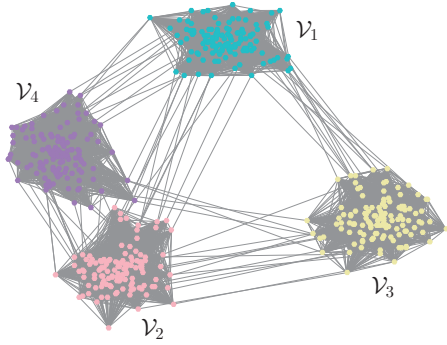}
    \vspace{-1em}
    \caption{Network topology for a $500$-node network.}
    \label{clustplot}
    \vspace{-1.5em}
\end{figure}

In this section we verify the performance of our proposed designs using a first-order consensus network. The plant $G_{22}$ is described by $n_{s}=500$ integrators that are interconnected by the network topology shown in Fig. \ref{clustplot}, where each node represents a single integrator. The dynamics at each integrator is modeled by 
\begin{align*}
\dot{x}_{i} = u_{i} + \sum_{j=1,j\neq i}^{n_{s}} a_{ij}(x_{j} - x_{i}), \quad i=1,...,n_{s},
\end{align*}
where $u_{i}$ is a scalar control input, and $a_{ij}\geq 0$ represents the connection strength between the $i^{th}$ and $j^{th}$ oscillators. We assume the state $x_{i}$ to be measurable, i.e., $y_{i} = x_{i}$. The design parameters $C_{1}$ and $B_{1}$ are chosen as identity matrices scaled by $10$, and $D_{12}$ and $D_{21}$ are chosen as identity matrices. For the hierarchical constraint $\mathcal{S}$, the weight vectors $(w_{u},w_{y})$ are chosen as vectors of all ones, which allow $P_{u}$ and $P_{y}$ to satisfy the stabilizability and detectability conditions in Theorem \ref{t1}. The clustering sets $\mathcal{I}^{u}$ and $\mathcal{I}^{y}$ are selected as four groups of integrators indexed by $\mathcal{V}_{1} = \{ 1,...,125 \}$, $\mathcal{V}_{2} = \{ 126,...,250 \}$, $\mathcal{V}_{3} = \{ 251,...,375 \}$ and $\mathcal{V}_{4} = \{ 376,...,500 \}$, where inside the groups integrators are densely connected while the groups themselves are sparsely connected, as shown in Fig. \ref{clustplot}. This results in four naturally coherent clusters. Given this constraint $\mathcal{S}$, the optimal hierarchical controller $K_{opt}$ requires $3.6418$ seconds to solve, and yields an $\mathcal{H}_{2}$ norm $\| f(G,K_{opt})\|_{\mathcal{H}_{2}} = 2.3191$. 

We first verify the Hamiltonian-based approximation with respect to choices of $\kappa$ from $1$ to $6$, and plot the corresponding computation costs and $\mathcal{H}_{2}$ norms yielded by the approximated controllers in Fig. \ref{hamiplot}. For comparison, $\mathcal{H}_{2}$ norms shown on the right axis are normalized/divided by $2.3191$ from the optimal case. As $\kappa$ increases, the normalized $\mathcal{H}_{2}$ norm decreases, and becomes close to $1$ when $\kappa \geq 4$, indicating a close performance matching between the approximated controller and $K_{opt}$. On the left axis, the computation cost for the approximated controller scales up as $\kappa$ increases. When $\kappa \leq 4$, the computation costs are all under $0.3$ seconds, which are much lower compared with $3.6418$ seconds for $K_{opt}$. Therefore, at the intersection point $\kappa =4$, one can design an approximated controller in a much lower cost, yet still get the similar performance as $K_{opt}$. This performance is mainly facilitated by the $4$ smallest dominant eigenvalues resulting from the natural clustering of the open-loop network shown in Fig. \ref{clustplot}. The $4$ dominant eigenvalues also explain the sudden increase in computation time at $\kappa =5$, where the convergence of Arnoldi algorithm become slower in searching the one extra non-dominant eigenvalue. It is worth mentioning that if the open-loop system is not naturally clustered, the $\mathcal{H}_{2}$ norm can increase. The computation cost, however, will not change significantly.
\begin{figure}
    \centering
    \includegraphics[width=0.9\columnwidth]{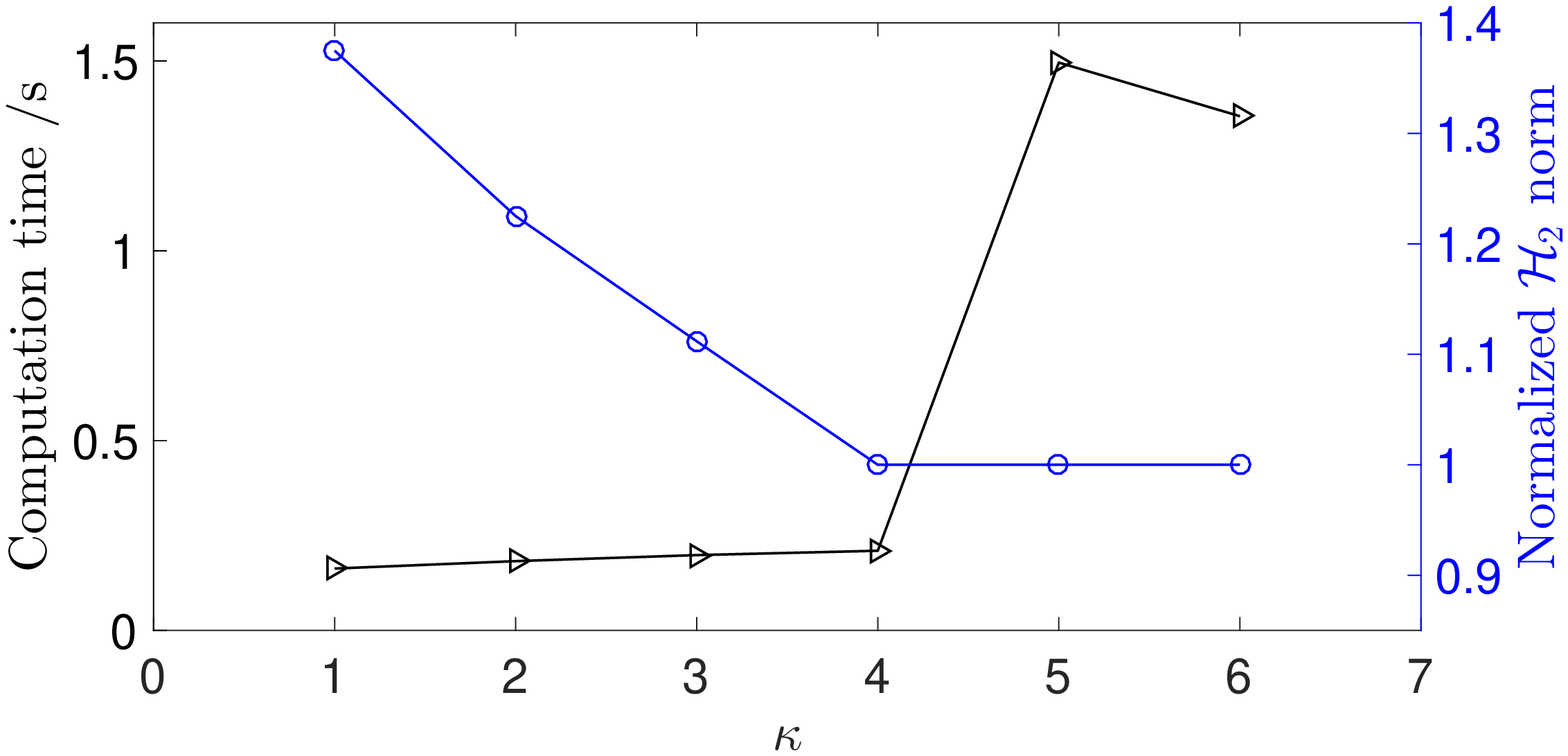}
    \vspace{-0.5em}
    \caption{Performance results for Hamiltonian-based reduction.}
    \label{hamiplot}
    \vspace{-0.5em}
\end{figure}
\begin{figure}
    \centering
    \includegraphics[width=0.9\columnwidth]{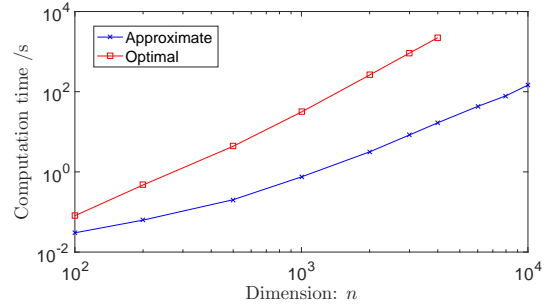}
    \vspace{-0.5em}
    \caption{Scalability results for computation time}
    \label{scala}
    \vspace{-1.5em}
\end{figure}
\begin{figure}
\vspace{-1em}
    \centering
    \includegraphics[width=0.9\columnwidth]{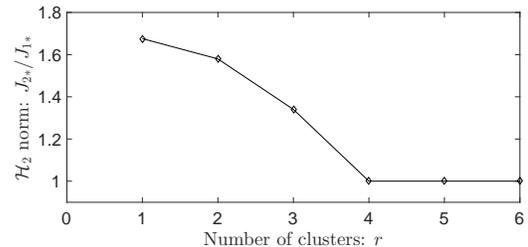}
    \vspace{-0.5em}
    \caption{Optimality results for design of $\mathcal{S}$}
    \label{Splot}
\end{figure}

The scalability of our design is further tested upon consensus networks whose dimensions range from $100$ to $10000$ with $a_{ij}$ generated randomly. The resulting computation costs for solving the Hamiltonian-based approximation versus that of $K_{opt}$ are plotted in Fig. \ref{scala}. As $n$ scales up, the approximate approach becomes remarkably cheaper than the optimal case. For example, at $n=4000$, solving the optimal case is already beyond our computing capacity, while the approximate approach can be completed within only $17$ seconds. This verifies the tractability of our design in handling large-scale networks.

We finally illustrate the design of $\mathcal{S}$ for the same $500$-node network. The clustering sets $(\mathcal{I}^{u},\mathcal{I}^{y})$ in this example are determined from (\ref{kmeansop}) with respect to $r=1,...,6$, and then used to design the controller $K_{opt}$. We show in Fig. \ref{Splot} the ratio between $J_{2*}$ and $J_{1*}$, which as defined in Section \RNum{5} are the resulting optimal values of constrained and unconstrained $\mathcal{H}_{2}$ problems respectively. The ratio $J_{2*}/J_{1*}$ as shown in Fig. \ref{Splot} decreases monotonically, and becomes close to $1$ at $r=4$. The clusters generated by $r=4$ actually resembles the four coherent groups in Fig. \ref{clustplot}.  This verifies our proposed design in finding a hierarchical structure that tightens the gap between constrained and unconstrained $\mathcal{H}_{2}$ problems.

\section{Conclusion}
In this paper we presented a hierarchical $\mathcal{H}_{2}$ optimal controller that provides benefits for both design and implementation for large-scale network dynamic systems. The advantage of the hierarchical structure is that it satisfies quadratic invariance property for any generic network model, and allows a convex reformulation for the original problem. The reformulated design can be simplified by conventional controller reduction techniques without breaking the hierarchical structure of the controller. In addition, we proposed a preliminary design of the clustering sets to tighten the gap between unconstrained and constrained $\mathcal{H}_{2}$ problems. Our future work will be to extend this design to the clustering weights.

\appendix

\subsection*{Proof of Theorem \ref{theob}}

By applying the definition of $\mathcal{H}_{2}$ norm, the objective function in (\ref{ophm}) can be rewritten as
\begin{align*}
\| (X-\bar{X})W_{x}\|_{\mathcal{H}_{2}} = \| (X-\bar{X})\Phi^{\frac{1}{2}} \|_{F},
\end{align*}
where $\Phi = \Phi^{\frac{1}{2}}\Phi^{\frac{T}{2}} = \mathrm{LYAP}(A-MX,B_{1})$. Note that from (\ref{hamieig}) we can find that $A-MX = Z_{1}\Lambda^{-}Z_{1}^{-1}$, which means $\Lambda^{-}$ contains all the eigenvalues of closed-loop state matrix $A-MX$ on the eigenspace $Z_{1}$. Hence, the controllability Gramian $\Phi$ follows directly from \cite{antoulas} as $\Phi = Z_{1}\mathcal{C}Z_{1}^{T}$. Also given the expansion of the error matrix $(X-\bar{X})$ as
\begin{align*}
(X-\bar{X}) = \begin{bmatrix}
0 & Z_{\bar{\kappa}} - Z_{2\kappa}(Z_{2\kappa}^{T}Z_{1\kappa})^{-1}Z_{2\kappa}^{T}Z_{1\bar{\kappa}}
\end{bmatrix} Z_{1}^{-1},
\end{align*}
the norm $\| (X-\bar{X})\Phi^{\frac{1}{2}}\|_{F}$ follows
\begin{align*}
& \| (X-\bar{X})\Phi^{\frac{1}{2}}\|_{F} = \sqrt{tr[(X-\bar{X})\Phi (X-\bar{X})]} \\
&\  = \sqrt{tr(E_{\kappa} \mathcal{C}_{\kappa+1:n,\kappa+1:n} E_{\kappa}^{T})} = \| E_{\kappa} \mathcal{C}_{\kappa+1:n,\kappa+1:n}^{\frac{1}{2}} \|_{F} \\
&\  \leq \| E_{\kappa} \|_{F} \| \mathcal{C}_{\kappa+1:n,\kappa+1:n}^{\frac{1}{2}} \|_{F} = \| E_{\kappa} \|_{F} \sqrt{tr(\mathcal{C}_{\kappa+1:n,\kappa+1:n})}. 
\end{align*}
Expanding the trace above then yields (\ref{mainbound}).

\subsection*{Proof of Lemma \ref{equiv}}

Given that $P_{u}$ and $P_{y}$ are orthonormal matrices, it can be verified from PBH test that $(A,B_{2}P_{u}^{T}P_{u})$ is stabilizable and $(P_{y}^{T}P_{y}C_{2},A)$ is detectable. Following the same rationale in the proof of Theorem \ref{t1}, $J_{2}$ can be rewritten as $J_{2} =  f(\breve{G},\breve{K}) $, where $\breve{G}$ follows 
\begin{align*}
\breve{G} = \left[
\begin{array}{c|cc}
A & B_{1} & B_{2}P_{u}^{T}P_{u} \\ \hline
C_{1} & 0 & D_{12}P_{u}^{T}P_{u} \\
P_{y}^{T}P_{y}C_{2} & P_{y}^{T}P_{y}D_{21} & 0
\end{array}
\right]
\end{align*}
and $\breve{K}$ is any stabilizing controller for $\breve{G}$ parameterized by Theorem \ref{para}. An important property is that the inversion of the matrix $P_{u}D_{12}^{T}D_{12}P_{u}^{T}$ (or $P_{y}D_{21}D_{21}^{T}P_{y}^{T}$) follows
\begin{align*}
(P_{u}D_{12}^{T}D_{12}P_{u}^{T})^{-1} = P_{u}(P_{u}^{T}P_{u}D_{12}^{T}D_{12}P_{u}^{T}P_{u})^{-1}P_{u}^{T}.
\end{align*}
Given this equality, we can find that
\begin{align*}
\mathrm{ARE}(A,B_{2}P_{u}^{T}P_{u}, C_{1}, P_{u}^{T}P_{u}D_{12}^{T}D_{12}P_{u}^{T}P_{u}) = X, \\
\mathrm{ARE}(A^{T},C_{2}^{T}P_{y}^{T}P_{y}, B_{1}^{T}, P_{y}^{T}P_{y}D_{21}D_{21}^{T}P_{y}^{T}P_{y}) = Y,
\end{align*}
from which it can be verified that the state space solution of (\ref{ope}) is same as $K_{opt}$ from Theorem \ref{t1}. Therefore, (\ref{ope}) is equivalent to (\ref{op3}).

\subsection*{Proof of Theorem \ref{boundr}}

Denote the complement of matrices $P_{u}$ and $P_{y}$ by $\bar{P}_{u}$ and $\bar{P}_{y}$ respectively, the function $J_{1}(Q)$ can be written as
\begin{align*}
J_{1}(Q) = T_{11} + T_{12}(P_{u}^{T}P_{u} + \bar{P}_{u}^{T}\bar{P}_{u})Q(P_{y}^{T}P_{y} + \bar{P}_{y}^{T}\bar{P}_{y})T_{21}.
\end{align*}
By letting $Q=Q_{*}$, we can write 
\begin{align*}
\| J_{2}(Q_{*})\|_{\mathcal{H}_{2}}^{2} = J_{1*}^{2} + \| J_{e}\|_{\mathcal{H}_{2}}^{2} + 2 \| J_{e}^{\frac{1}{2}}J_{1}^{\frac{1}{2}}(Q_{*}) \|_{\mathcal{H}_{2}}^{2},
\end{align*}
where $J_{e} {=} T_{12}P_{u}^{T}P_{u}Q_{*}\bar{P}_{y}^{T}\bar{P}_{y}T_{21} {+} T_{12}\bar{P}_{u}^{T}\bar{P}_{u}Q_{*}P_{y}^{T}P_{y}T_{21} + T_{12}\bar{P}_{u}^{T}\bar{P}_{u}Q_{*}\bar{P}_{y}^{T}\bar{P}_{y}T_{21}$. Note that $J_{e} = 0$ when $\xi_{u} = \xi_{y} = 0$, or equivalently $\bar{P}_{u}^{T}\bar{P}_{u}Q_{*} = Q_{*}\bar{P}_{y}^{T}\bar{P}_{y} = 0$. As a result, the equation above yields the equality $J_{1*} = \| J_{2}(Q_{*})\|^{2}_{\mathcal{H}_{2}} = J_{2*}$. The bound in (\ref{bound}) simply follows the inequality
\begin{align*}
J_{2*}^{2} & \leq \| J_{2}(Q_{*})\|_{\mathcal{H}_{2}}^{2} = J_{1*}^{2} + \| J_{e}\|_{\mathcal{H}_{2}}^{2} + 2 \| J_{e}^{\frac{1}{2}}J_{1}^{\frac{1}{2}}(Q_{*}) \|_{\mathcal{H}_{2}}^{2} \\
& \leq J_{1*}^{2} + \| J_{e}\|_{\mathcal{H}_{2}}^{2} + 2J_{1*}\| J_{e}\|_{\mathcal{H}_{2}} \leq J_{1*}^{2} + \xi^{2} + 2\xi J_{1*} . 
\end{align*}

\end{document}